\documentclass[10pt,twocolumn,twoside]{IEEEtran}
\IEEEoverridecommandlockouts

\usepackage{cite}
\usepackage{amsmath,amssymb,amsfonts}
\usepackage{algorithmic}
\usepackage{graphicx}
\usepackage{textcomp}
\usepackage{balance}
\usepackage[dvipsnames]{xcolor}
\usepackage[utf8]{inputenc}
\usepackage[most]{tcolorbox} % for tcolorbox, breakable, enhanced
\usepackage{enumitem}        % for customized enumerate labels (label=\arabic*., \alph*.)
\usepackage{xcolor}          % color support (usually loaded automatically, but safe)
\usepackage[breaklinks,colorlinks, linkcolor=MidnightBlue, anchorcolor=MidnightBlue, citecolor=MidnightBlue, urlcolor=MidnightBlue]{hyperref}
\usepackage[hyphenbreaks]{breakurl}
\usepackage{cite}
\usepackage{xurl} 
\usepackage{bbm}

\newtheorem{definition}{Definition}
\newtheorem{proposition}{Proposition}% 
\newtheorem{proof}{Proof}
\def\BibTeX{{\rm B\kern-.05em{\sc i\kern-.025em b}\kern-.08em
    T\kern-.1667em\lower.7ex\hbox{E}\kern-.125emX}}
\usepackage{lipsum}

\title{\LARGE \bf 
Behavioral Generative Agents for Power Dispatch and Auction 
}
\author{Shaoze Li$^\dagger$, Justin S. Kim$^\dagger$, Cong Chen$^\dagger$% <-this % stops a space
\vspace{-0.5cm}
\thanks{\scriptsize
This work is under review. ~\textit{(Corresponding author: Cong Chen.)}  }
 \thanks{ \scriptsize
$^\dagger$ Authors are with the Thayer School of Engineering, Dartmouth College, Hanover, NH, USA. Email: {\tt Cong.Chen@dartmouth.edu} {\tt Shaoze.Li.TH@dartmouth.edu} {\tt Justin.S.Kim.TH@dartmouth.edu}. The research was supported by Amazon Research Award (Spring 2025), and OpenAI Researcher Access Program. Any opinions, findings, and conclusions or recommendations expressed in this material are those of the author(s) and do not reflect the views of funding agencies.}
}

\begin{document}
\begingroup
\allowdisplaybreaks

\maketitle

\begin{abstract}
%Designing power dispatch and market mechanisms that accurately reflect diverse customer preferences remains a critical challenge for future energy systems. This paper provides positive initial evidence for the feasibility of large language model (LLM)–powered generative agents in power dispatch and auction applications through two proof-of-concept experiments. First, we construct a home battery management testbed with stochastic prices and blackout interventions, and benchmark LLM decisions against dynamic programming and greedy policies. By incorporating an in-context learning (ICL) module, we show that behavioral patterns extracted from a stronger reasoning model can be transferred to a smaller LLM via example-based prompting, leading agents to prioritize post-blackout energy reserves over short-term profit. Furthermore, in auction environments, agents equipped with structured reasoning frameworks exhibit complex behavioral heterogeneity, evolving from myopic rule-followers to aggressive, strategic bidders aiming for long-term dominance. Our findings indicate that behavioral generative agents offer a vital bridge between theoretical optimization and practical implementation, serving as a scalable testbed for evaluating customer-centric energy policies and market designs.
This paper presents positive initial evidence that generative agents can relax the rigidity of traditional mathematical models for human decision-making in power dispatch and auction settings. We design two proof-of-concept energy experiments with generative agents powered by a large language model (LLM). First, we construct a home battery management testbed with stochastic electricity prices and blackout interventions, and benchmark LLM decisions against dynamic programming. By incorporating an in-context learning (ICL) module, we show that behavioral patterns discovered by a stronger reasoning model can be transferred to a smaller LLM via example-based prompting, leading agents to prioritize post-blackout energy reserves over short-term profit. Second, we study LLM agents in simultaneous ascending auctions (SAA) for power network access, comparing their behavior with an optimization benchmark, the straightforward bidding strategy. By designing ICL prompts with rule-based, myopic, and strategic objectives, we find that structured prompting combined with ICL enables LLM agents to both reproduce economically rational strategies and exhibit systematic behavioral deviations. Overall, these results suggest that LLM-powered agents provide a flexible and expressive testbed for modeling human decision-making in power system applications.

\end{abstract}


%%%%%%%%%%%%%%%%%%%%%%%%%%%%%%%%%%%%%%%%%%%%%%%%%%%%%%%%%%%%%%%%%%%%%%%%%%%%%%%%
\section{Introduction}

Recent advances in generative agents powered by large language models (LLMs) have sparked growing interest in their potential to simulate human decision-making (“homo silicus”) across a wide range of applications, including education, experimental design, and socio-technical systems \cite{park25GenerativeAgents}. Despite this promise, current LLMs exhibit important limitations in real-world settings, such as hallucinations, prompt misinterpretation, and output variability. These shortcomings constrain the use of LLM-based agents as faithful human simulators. Nevertheless, rapid progress in LLM capabilities motivates systematic exploration of their practical value under carefully designed experimental settings.

In this work, we investigate how LLM-based agents can be leveraged to (1) represent diverse behavioral personas and (2) contribute to decision-making in power system dispatch and auction mechanisms. Rather than aiming for fully realistic human simulations—which remain challenging given current LLM limitations—we focus on simplified yet representative experiments that reduce ambiguity while still capturing meaningful decision complexity. This approach allows us to evaluate whether LLM agents can demonstrate nontrivial decision sophistication in energy system contexts where human behavior plays a critical role.

We develop a behavioral generative agent testbed and benchmark LLM-based decisions against classical mathematical decision-making models. Through this comparison, we assess the extent to which LLM agents can explore strategy spaces that are difficult to represent using rigid optimization frameworks. We present two proof-of-concept experiments—one in power dispatch and one in auctions—to illustrate the potential of LLM agents as an intermediate evaluation tool between purely analytical models and costly human-subject experiments.

Traditionally, energy operation models rely on optimization frameworks that assume fully rational agents maximizing a well-defined objective, such as cost minimization or utility maximization. In practice, however, energy customers and operators often exhibit heterogeneous preferences, bounded rationality, and context-dependent behavior. Recent advances in generative AI offer a promising pathway to incorporate such behavioral diversity into energy system analysis. By conditioning LLMs on different personas, preferences, and decision histories, generative agents can approximate a broader range of human decision-making patterns. Related work can be broadly categorized into two strands.

\paragraph{Generative agents for auctions and strategic operations}
A growing literature treats LLMs as simulated economic agents capable of participating in repeated strategic environments, including auction-like settings, through natural-language interaction \cite{horton2023large,aher2023using,ArgyleEtAl2023OutOfOneMany}. These agents often display systematic deviations from strict utility maximization, enabling the study of behavioral biases relevant to strategic choice, such as mental accounting and distorted risk preferences \cite{RossKimLo2024LLMEconomicus,Leng2024MentalAccounting}. LLM agents have been evaluated in repeated games and multi-agent benchmarks closely related to auction dynamics, revealing both promise and consistent limitations in strategic reasoning \cite{akata2025playing,fontana2025nicer,duan2024gtbench}. A key methodological insight is that LLMs are stateless: their behavior depends critically on how the evolving game state and history are encoded in the prompt. Prior work shows that alternative textual state representations can substantially alter learning dynamics and outcomes \cite{GoodyearGuoJohari2025StateRep}. As a result, in-context learning (ICL)—conditioning on examples, past interactions, and feedback without parameter updates—plays a central role in enabling adaptive behavior \cite{brown2020language,wei2022chain}. In this setting, prompt design functions as an implicit learning mechanism rather than a mere interface.

\paragraph{Generative agents for power dispatch and auctions}
In contrast, the literature on energy systems has traditionally emphasized optimization-based modeling and equilibrium analysis, often relying on simplified or calibrated behavioral assumptions \cite{AgrawalYucel2022DRPrograms}. Recent work in operations research examines the  strategic behavior of distributed energy resources (DERs) and market design implications in modern power systems \cite{GaoAlshehriBirge2024DERMarketPower}. Bridging these perspectives, Chen et al.\ introduce behavioral generative agents for energy operations, using LLM prompting to emulate heterogeneous consumer personas and generate both decisions and natural-language rationales under dynamic prices and rare events \cite{Chen25behavioral}. This direction aligns with the broader homo silicus view of artificial agents as computational proxies for economic behavior \cite{FilippasHortonManning2024HomoSilicusEC} and with general-purpose agent frameworks that operationalize preferences and narratives through prompts \cite{ShanahanMcDonellReynolds2023RolePlay}.

% Crucially, these energy-agent studies often \emph{do use ICL}: the policy is induced by conditioning the LLM on (i) the current operational state (e.g., prices, storage state-of-charge, constraints), (ii) selectively summarized histories of prior days, and (iii) decision rubrics/examples, without any parameter updates—making context-window construction and summarization a first-order modeling choice \cite{GoodyearGuoJohari2025StateRep,LiuEtAl2024LostInMiddle,BrownEtAl2020FewShot}.

Building on this literature, this paper provides positive initial evidence for the feasibility of LLM-powered behavioral generative agents in energy system applications through two proof-of-concept studies.

First, we develop a simplified home battery management testbed with stochastic electricity prices and blackout interventions, and benchmark LLM decisions against dynamic programming (DP) and greedy policies. By extending the agent framework with an ICL module, we show that behavioral patterns discovered by a stronger reasoning model can be transferred to a smaller LLM through example-based prompting. In particular, when provided with blackout-related ICL examples, the LLM agent learns to fully discharge the battery during blackout events and retain higher energy reserves afterward, deviating from purely profit-maximizing DP behavior. These effects are more pronounced for personas that prioritize reliability over revenue.

Second, we explore the use of LLM agents to simulate bidding behavior in power network access auctions, relevant for distributed energy resources and data centers seeking grid withdrawal rights. Based on the simultaneous ascending auctions (SAAs), we evaluate a \emph{Straightforward Bidding Strategy} as a baseline against three distinct LLM-based agents: a \emph{Rule-Centric agent} adhering to fundamental auction rules; a \emph{Myopic-Profit agent} focused on immediate next round profit; and a \emph{Strategic-Outcome agent} aiming for long-term auction gain. From the simulation results, we drive some interesting results: The bidding behavior of the \emph{Myopic-Profit Agent} aligns closely with the \emph{Straightforward Bidding Strategy}, since they have the same bidding objective, and the \emph{Strategic-Outcome Agent} exhibits a more aggressive bidding profile relative to both the \emph{Myopic-Profit Agent} and the \emph{Straightforward Bidding Strategy}, since the agent utilizes higher bids in earlier rounds to secure a dominant position for its desired products. Our comparative analysis demonstrates that incorporating structured Chain-of-Thought (CoT) and Thought–Action–Reflection–Journal (TARJ) prompting with ICL modules enables LLM agents to exhibit both higher rationality and behavioral biases in SAA, depending on specific objectives and conditions for each agent. 

% effectively mitigating over-aggressive bidding and improving overall market efficiency.}

Overall, our results suggest that LLM-powered generative agents offer a promising testbed for simulating human decision-making in power dispatch and auction mechanisms. While our experiments are intentionally simplified and remain at a proof-of-concept stage, they highlight the potential of LLM agents to support customer-centric energy management design, energy policy evaluation using population-level human simulators, and future human–AI collaboration in power system operations.

\section{Dispatch Behavior for Battery Operations}

 \begin{figure*}[ht!]
% \vspace*{-10 pt}%
\centering
\includegraphics[width=120mm]{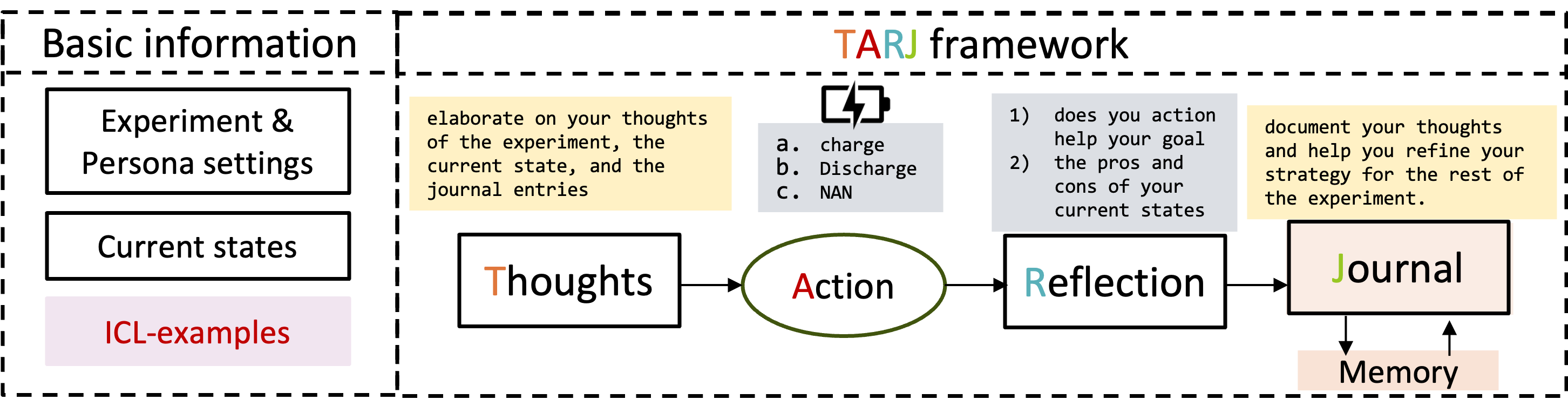}
\caption{LLM prompt design using a Thought-Action-Reflection-Journal (TARJ), augmented by in-context learning (ICL) examples.}
 \label{fig:TARJ-ICL}
\vspace*{-10 pt}%
\end{figure*} 
%  \vspace{-2em}
We study a simplified home battery management problem in which daily electricity prices fluctuate stochastically and charging or discharging decisions must be made under uncertainty. We benchmark LLM-based behavioral agents against an optimal DP policy and a heuristic greedy policy. Our objective is to evaluate whether LLM agents, conditioned on different personas and in-context examples, exhibit meaningful and interpretable deviations from rational optimization benchmarks.
 
To ensure tractability and reduce ambiguity for LLM-based decision-making, we model battery operations at a daily resolution with discrete actions—charge, discharge, or hold as shown among the action options in Fig.\ref{fig:TARJ-ICL}. Electricity prices are driven by simplified stochastic processes. %These abstractions allow direct comparison with DP-based optimal policies while retaining key dynamic power dispatch features under price uncertainties.

%We adopt a stochastic simulation framework with ICL and memory to guide LLM decisions. The believability of LLM-based behavioral agents is motivated by prior work on generative agents with memory and chain-of-thought (CoT) reasoning \cite{park25GenerativeAgents}.

\subsection{Agent decision dynamics with blackout intervention}

We extend the TARJ behavioral generative agent framework proposed in \cite{Chen25behavioral} by introducing an explicit ICL module. The overall prompt structure is illustrated in Fig.~\ref{fig:TARJ-ICL}. The prompt consists of two components: (i) basic information, including experiment settings, persona definitions, current electricity price, and battery state-of-charge (SoC); and (ii) the TARJ framework, which provides structured reasoning guidance.
  
Within TARJ, the LLM is prompted to first reason about the decision (Thought), select a battery action (Action), analyze the outcome (Reflection), and summarize insights for future decisions (Journal). Journal entries are stored in memory and dynamically incorporated into subsequent prompts. This memory mechanism enables the agent to adapt behavior over time, consistent with prior findings on generative agents \cite{park25GenerativeAgents}.
 
The simulation proceeds sequentially over time, repeatedly invoking the TARJ–ICL prompt in Fig.~\ref{fig:TARJ-ICL} while updating memories in the dynamic power dispatch experiment. This dynamic testbed enables controlled interventions, including a one-day blackout event. Such low-frequency, high-impact events are difficult to study using traditional human-subject experiments but can be naturally incorporated into LLM simulations.

\subsection{In-context learning for behavior discovery}\label{sec:ICL}

In-context learning (ICL) refers to an LLM’s ability to adapt its behavior based on the examples provided in the prompt, without updating the model parameters. In our ICL module (Fig.~\ref{fig:TARJ-ICL}), we provide example power dispatch scenarios to guide behavioral patterns in the prompt.

A key design choice is to use decisions generated by a more advanced reasoning model (o1-preview) during blackout interventions as ICL examples for a smaller and cheaper model (gpt-5-mini). We find that behavioral patterns discovered by o1-preview—such as prioritizing energy reserves after blackout exposure—can be effectively transferred to gpt-5-mini through ICL. This transfer is enabled by the linguistic representation of reasoning, which is easier for LLMs to imitate than numerical policies. Detailed quantitative results are reported in the simulation result from Fig.~\ref{fig:SoCall}.

\subsection{Dynamic program policies without blackout intervention}

When blackout interventions are excluded, the battery dispatch problem admits a standard rational optimization formulation. We consider two mathematical benchmarks: an optimal DP policy and a heuristic greedy policy.

An ideal rational consumer solves the following stochastic optimization problem:
\begin{equation}\label{eq:expectedU}
    \max_{u_t \in \mathcal{U}} ~\mathbb{E}\left[\sum_{t=1}^{T} \pi_t u_t \right] \quad \text{s.t.}~  s_{t+1} = s_{t} - u_t, ~\underline{s} \leq s_{t} \leq \overline{s},~\forall t \in [T].
\end{equation}
Here, $\pi_t \in \mathbb{R}$ denotes the electricity price, $u_t \in \mathbb{R}$ the battery action (positive for discharging, negative for charging), and $s_t$ the battery SoC at time $t$. The parameters $\underline{s}$ and $\overline{s}$ represent SoC limits. This DP benchmark provides the optimal arbitrage policy under rational expectations.

Although exact DP solutions can be computationally intensive \cite{bertsekas2012dynamic}, they serve as a gold standard. We also implement a greedy heuristic that charges at low prices and discharges at high prices \cite{JiangPowell15Storage}. This heuristic reflects bounded rationality and is expected to be suboptimal relative to DP.

\subsection{Experimental Settings}

We simulate a 20-day home battery operation scenario with the following settings: 1) energy customer can take daily action to charge the battery by 1 kWh (buy energy), discharge the battery by 1 kWh (sell energy) to the grid, Do nothing; 2) the battery can hold up to 10 kWh of charge; 3) The daily price is simplified to be either \$10/kWh or \$5/kWh, with a 50\% chance for each value; 4) At the end of the 20-day experiment, the battery will be removed, and any unused energy left in the battery will not be compensated.

Each day, gpt-5-mini is queried to select an action based on the current price, SoC, ICL examples, and memory. ICL examples are generated by o1-preview, while all simulations are executed using gpt-5-mini. On day 10, a one-day blackout intervention is introduced, during which the battery may fully discharge to serve household demand.

We consider three personas—Thinker, Realist, and Feeler—with distinct preferences. Persona definitions and full prompt details are provided in Appendix~\ref{sec:batteryexperiment}. For each persona, we run 40 Monte Carlo simulations to capture both price uncertainty and LLM output variability.

We compare LLM behaviors against the DP and greedy benchmarks, shown as dashed lines in Fig.~\ref{fig:SoCall}. Deviations from the DP policy reveal behavioral preferences not captured by rational optimization.

\begin{figure}
    \centering
    \includegraphics[width=0.45\linewidth]{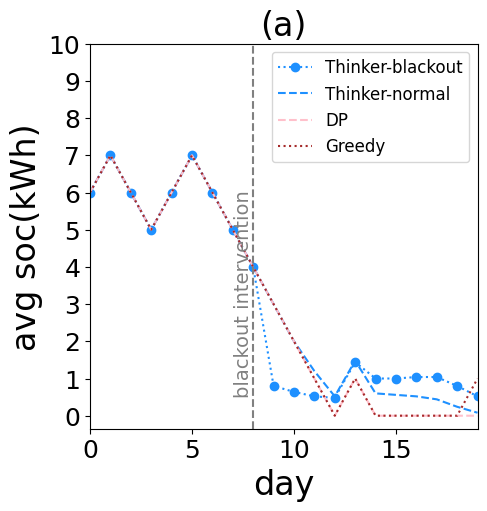}\includegraphics[width=0.45\linewidth]{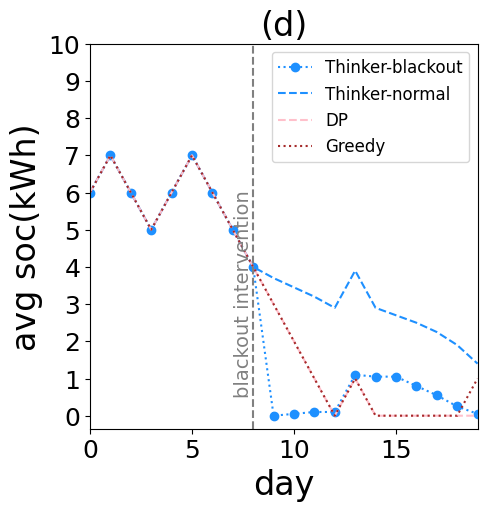}
    \includegraphics[width=0.45\linewidth]{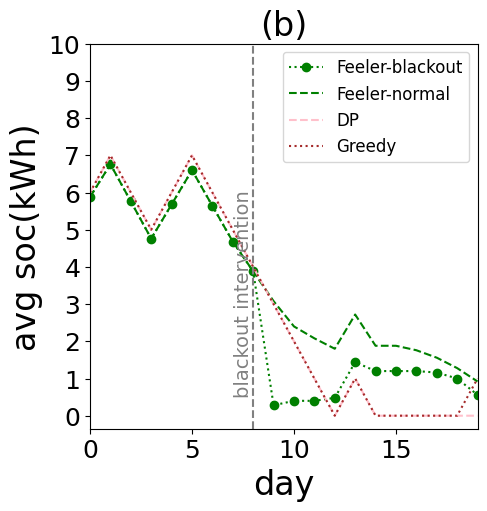}\includegraphics[width=0.45\linewidth]{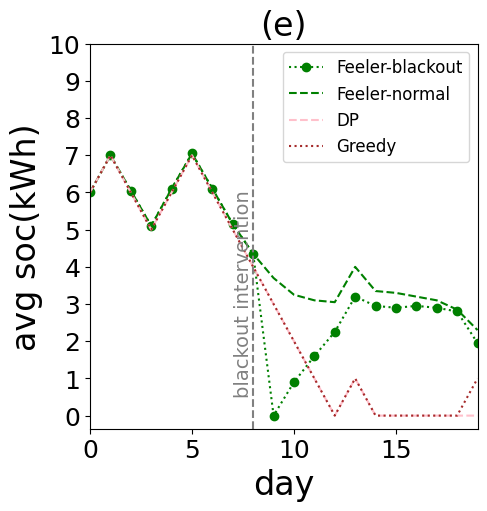}
    \includegraphics[width=0.45\linewidth]{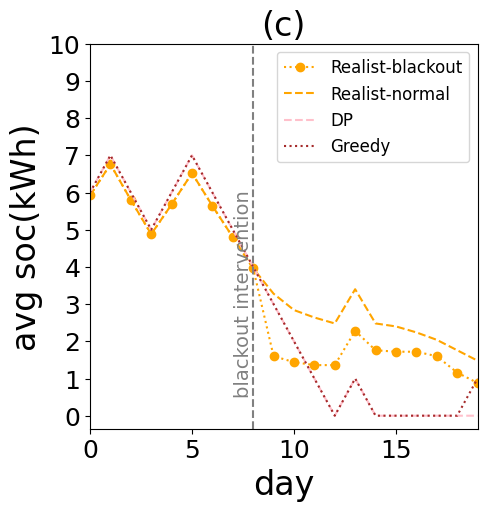}\includegraphics[width=0.45\linewidth]{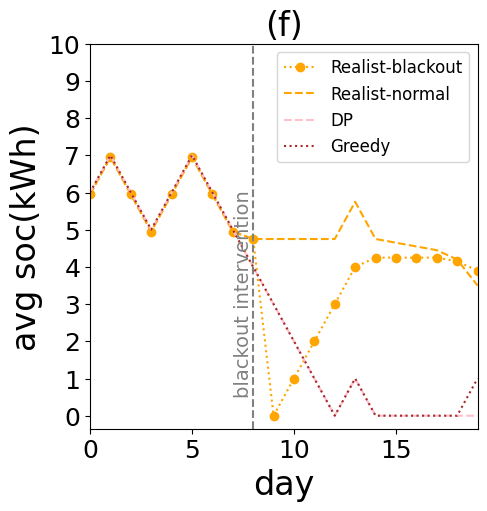}
    \caption{Comparing SoC resulted from LLM decision during blackout. Left: without IC examples. Right: with ICL-blackout examples.}
    \label{fig:SoCall}
      \vspace{-1.6em}
\end{figure}

 \subsection{Experimental Results} \label{sec:simbattery}
 Our main finding is that gpt-5-mini successfully learns blackout-related behavioral patterns through ICL examples generated by o1-preview.

First, after being prompted with blackout ICL examples, gpt-5-mini simulates Feeler and Realist personas that retain higher terminal SoC, reflecting a preference for backup energy reserves. Without ICL examples (left column of Fig.~\ref{fig:SoCall}), terminal SoC is generally low, consistent with profit-maximizing behavior under \eqref{eq:expectedU}. With ICL examples (right column), terminal SoC increases for these personas, matching the behavior discovered by o1-preview (Fig.~\ref{fig:SoCo1} in the Appendix).

Second, gpt-5-mini learns to fully discharge the battery during the blackout event across all personas. Without ICL examples, average SoC during the blackout day remains positive. With ICL examples, average SoC drops to zero, consistent with o1-preview behavior.

We also observe increased behavioral heterogeneity across personas when ICL examples are provided. This suggests that diverse ICL examples improve the ability of smaller models to represent distinct personas more clearly.

Finally, we note limitations of ICL-based transfer. While gpt-5-mini captures qualitative patterns from o1-preview, quantitative differences remain. For example, o1-preview typically produces lower terminal rewards and higher terminal SoC for the Feeler persona, whereas gpt-5-mini shows the opposite trend, albeit with small differences. This highlights the bounds of behavior transfer via ICL.

\section{Bidding Behavior for Power Auctions}

Power auctions encompass human bidding decisions with diverse objectives and information processing procedures. In this study, we specifically examine power network access auctions, where data centers compete for grid withdrawal rights, to evaluate the performance of LLM agents in complex market environments. We introduce three distinct agent architectures, each possessing different levels of informational and strategic awareness: \emph{Rule-Centric Agent, Myopic-Profit Agent, Strategic-Outcome Agent}. The performance of these agents is benchmarked against a classical \textit{Straightforward Bidding}  strategy.
Our objective
is to evaluate whether LLM agents will perform rationally and strategically.

The remainder of this section is organized as follows: first, we define the mechanisms of the SAA and the governing rules of the \textit{Straightforward Bidding}  strategy. Second, we detail the internal logic and prompting frameworks for our three LLM-based agents. Finally, we present simulation results comparing the LLM-derived bidding strategies against the baseline to demonstrate their comparative rationality and efficiency.

% \section{Organizational Behavior for Power Auction}

\subsection{Power Network Access Right Auction Mechanism}
We model the allocation of withdrawal capacity as a SAA. Let $\mathcal{N} = \{1, \dots, N\}$ denote the set of bidders (technology firms) and $\mathcal{K} = \{1, \dots, K\}$ denote the set of items, where each item corresponds to a specific withdrawal capacity at a candidate location. The auction proceeds in discrete rounds indexed by $t = 1, 2, \dots$.

\subsubsection{Auction Rules and Price Updates}
At each round $t$, the auctioneer posts the standing high bid $H_{k, t-1} \in \mathbb{R}$ and the current standing high bidder $w_{k, t-1} \in \mathcal{N}$ for each item $k \in \mathcal{K}$. These are common information for all bidders. 

To participate in round $t$, a bidder $n$ must submit a sealed bid $b^{(n)}_{k,t}$ that meets the minimum bid requirement. Crucially, this requirement depends on whether the bidder currently holds the standing high bid for the item. 

% \textcolor{red}{[Cong: Try this $r_{k,t}^{(n)} = H_{k, t-1} + \delta_k\mathbbm{1}{\{n \neq w_{k, t-1}\}}$. Try this {\em minimum bidding price}]\ }
\begin{definition}
The minimum bidding price for bidder $n$ on item $k$, denoted by $r_{k,t}^n$, is defined as:

\begin{equation}
   r_{k,t}^{(n)} = H_{k, t-1} + \delta_k\mathbbm{1}{\{n \neq w_{k, t-1}\}}
\end{equation}
where $\delta_k > 0$ is a predetermined bid increment and $\mathbbm{1}\{\cdot\}$ is an indicator function.    
\end{definition}
This minimum bidding price allows the standing high bidder to maintain their position at the current price, while challenging bidders must increase the price to displace the incumbent.

The standing high bid for item $k$ at the end of round $t$ updates according to the highest eligible new bid:
\begin{equation}
    H_{k,t} = \max\left(H_{k, t-1}, \max_{n \in \mathcal{N}} \{b^{(n)}_{k,t}\}\right).
\end{equation}
If no new bids are placed on item $k$, the standing high bid remains unchanged, and the item is temporarily retained by the previous high bidder (i.e., $w_{k,t} = w_{k, t-1}$).

While the classic SAA framework typically incorporates \emph{Eligibility and Activity Rules} to ensure consistent bidding throughout the process, we have omitted these constraints in the primary study to simplify the decision space for the LLM agents. This simplification allows for a clearer evaluation of the agents' core strategic logic and reduces the potential for hallucination by excessive constraints.  We will consider these rules in future work

% For the sake of theoretical completeness, the full introduction about these rules is detailed in the Appendix.

\subsubsection{Straightforward Bidding Strategy}
We assume that bidders employ a \textit{Straightforward Bidding}  strategy. Under \textit{Straightforward Bidding}  strategy, a bidder chooses the maximal profits for the next round, ignoring the potential strategic impact of their bids on future prices. At each round $t$, bidder $n$ observes the current minimum bidding prices vector $\mathbf{r}^{(n)}_t := \{r^{(n)}_{1,t}, \dots, r^{(n)}_{K,t}\}$.

Let $v_{n,k}$ be the private valuation of bidder $n$ for location $k$. The bidder selects a bundle of items $S_{n,t}^*$ that maximizes their surplus subject to its minimum bidding prices:
\begin{equation}\label{p1}
    \max_{S \subseteq \mathcal{K}, b^{(n)}_{k,t}} \sum_{k \in S} (v_{n,k} - b^{(n)}_{k,t}) \quad \text{s.t.} \quad b^{(n)}_{k,t}\ge r^{(n)}_{k,t}, \forall k\in\mathcal{K}.
\end{equation}
% \textcolor{red}{[Cong: I suggest to write it as proposition rather than theorem. Only put a theorem if it's very difficult to prove. If you didn't use abbreviation SB anywhere else in the paper, replace Straightforward Bidding (SB) with Straightforward Bidding. ]}
\begin{proposition}
   Problem (\ref{p1}) is equivalent to the following  problem:
   \begin{equation}\label{p2}
    \max_{S \subseteq \mathcal{K}} \sum_{k \in S} (v_{n,k} - r^{(n)}_{k,t}).
\end{equation}
\end{proposition}
\begin{proof}
To maximize the surplus subject to the constraint $b^{(n)}_{k,t} \ge r^{(n)}_{k,t}$, the bidder must choose the smallest possible bid value allowed. Any bid $b^{(n)}_{k,t} > r^{(n)}_{k,t}$ would unnecessarily reduce the surplus, since the valuation $v_{n,k}$ is fixed. Therefore, for any chosen bundle $S$, the optimal bid for each item $k \in S$ is exactly the minimum requirement:    $b^{(n)*}_{k,t} = r^{(n)}_{k,t}, ~ \forall k \in S$.
Substituting $b^{(n)}_{k,t} = r^{(n)}_{k,t}$ into the objective function of \eqref{p1}, we obtain that $\sum_{k \in S} (v_{n,k} - r^{(n)}_{k,t})$. Consequently, maximizing the surplus in \eqref{p1} is equivalent to selecting the optimal bundle $S$ in \eqref{p2}.    
\end{proof}
The bidder then places bids on all items in the optimal bundle $S_{n,t}^*$ for which they are not currently the standing high bidder.

% {\color{red}Auction procedure and math model for anticipated results. Put simple simulatneous ascending auction here and you can comment that the power network constraints and probabilistic constraints will be considered in the future research.

% \textbf{Assumptions:} We assume the  

% \textbf{Key methodology:} While testing   }

\subsection{Generative Agent Testbed}

In this subsection, we will specify how to construct the LLM agent in this auction environment. Firstly, the module loads bidder identities, product specifications (e.g., node-specific withdrawal rights), and the private valuations for individual products and bundles assigned to each bidder. Global auction parameters—including starting prices and minimum bid increments ($\delta$)—are then established.

In each iteration of the SAA, the module invokes the LLM agents as active bidders. For the simulation environment, we utilize gpt-oss-120b, a model optimized for high-reasoning use cases, as the core engine for all bidding agents. Each agent is provided with a structured prompt that contains the current market state and their specific strategic objectives. We define three distinct agent architectures based on their decision-making horizons:

\begin{itemize}
    \item {\emph{Rule-Centric Agent}}: Operates strictly according to the fundamental auction rules of the auction. This agent acts as a basic model, submitting bids based solely on its private valuations and current price levels, without broader strategic consideration. The details are provided in Appendix \ref{sec:promptdesign} \textit{Prompt for the {Rule-Centric Agent}}. 
    
 \item{\emph{Myopic-Profit Agent}}: Focuses on maximizing immediate utility within the subsequent bidding round. Its primary objective is to secure the highest possible profit for the next round, adhering to auction rules while prioritizing short-term gains. The details are provided in Appendix \ref{sec:promptdesign} \textit{Prompt for Myopic-Profit Agent}. 

 \item{\emph{Strategic-Outcome Agent}}: Prioritizes long-term auction outcome and terminal profit. This agent focuses on future profits, and it may strategically forgo gains or accept losses in the current round to secure a more favorable outcome in the end. The details are provided in Appendix \ref{sec:promptdesign} \textit{Prompt for Strategic-Outcome Agent}. 

\end{itemize}

 In each round, the LLM agents generate a multi-faceted decision output comprising a bid vector, a target subset of products, and a structured reasoning sequence. This reasoning follows the strict TARJ format (Fig.~\ref{fig:TARJ-ICL}), which prompts the model to articulate its Thoughts, Action, Reflection, and Journal for the current iteration. This output serves a dual purpose:
 \begin{itemize}
     \item Internal State Management: The "Journal" component is fed back into the agent to maintain continuity and inform decision-making in subsequent rounds.
     \item  Market Clearing: The bid vector is processed by the auction mechanism, which validates bids and updates the price and provisional winner for each product.
 \end{itemize}

To resolve identical high bids, the mechanism employs a random tie-breaking rule. The auction concludes when a round yields no changes in prices or high bidders across all products. This termination criterion ensures the final allocation is a stable equilibrium under the prescribed increment-constrained bidding dynamics.

% \textcolor{red}{@Shaoze and Justin: When we embedded the short-term profit hint in the prompt, can I say that this is also one category of in-context learning (ICL)? Could you confirm this? If so, please write a subsection about in-context learning and how it's related to the prompt design. If not, then please briefly explain the difference between this prompt design and in-context learning. Draft of ICL is in Sec.~\ref{sec:ICL}.}

\subsection{In-Context Learning (ICL)}
As mentioned in Section \ref{sec:ICL}, ICL modules allow LLM models to modify their behavior based on examples and information in the prompt without updating model parameters. Here, by instructing the agents to use the journal to stay consistent with their own past strategy and inform their decision-making processes, the prompt turns these traces from the prior rounds into inference-time training signals, shaping the next bid without any parameter updates. With this ICL module, the prompt stabilizes the agents' behavior across rounds and enables consistent strategic adaptation under the same mechanism rules, making the LLM a more faithful approximation of a bidder with experience rather than a stateless one-shot responder. Furthermore, the prompts for \emph{Myopic-Profit} and \emph{Strategic-Outcome agents} force specific objectives, instructions, and contextual information for the agents, conditioning the agents to learn and formulate their bidding behaviors from context. This ICL module in the prompts, in return, encourages coherent inter-round planning, reduces erratic bidding, and makes the agents behave more like an adaptive economic agent whose policy evolves within the context of the ongoing auction.

\subsection{Simulation Results}
This experiment considers the auctioning of data center access rights in two geographical sites, denoted as A and B. We assume a duopolistic market structure featuring two high-tech enterprises. Their respective valuations for the access rights in each location are summarized in  Table \ref{tab1}. 
The numerical values presented in Table \ref{tab1} are expressed in millions. For instance, the value of bidder 1 for Product A corresponds to 4 million. For the sake of brevity, we omit the specific unit in the subsequent discussion. 

Using the valuations in Table \ref{tab1} and different prompts shown in Appendix~\ref{sec:promptdesign}, we run 30 Monte Carlo simulations for each agent to capture variability in both bids and LLM output. After the Monte Carlo simulations, we calculate the average bid results in each round and compare the bidding trajectories from the different agents with the \emph{Straightforward Bidding Strategy}. Our simulation results are shown in Fig.~\ref{fig:SAAsim}. Our observations of these simulations results are three-fold.

\begin{table}
\caption{Value Parameter Settings for SAA}
\begin{center}
\begin{tabular}{cccc}
\hline
\textbf{Table}&\multicolumn{3}{c}{\textbf{Bidders' true value of products}} \\
\cline{2-4} 
\textbf{Head} & {{A}}& {{B}}& {{A+B}} \\
\hline
Bidder 1& 4& 6 &  10\\

Bidder 2& 8& 4 &  12\\
\hline
\end{tabular}
\label{tab1}
\end{center}
  \vspace{-3em}
\end{table}

First, as illustrated in Fig.~\ref{fig:SAAsim}, the bidding behavior of the \emph{Myopic-Profit Agent} aligns closely with the benchmark: the \emph{Straightforward Bidding Strategy}.  It should be noted that the trajectories presented in Fig.~\ref{fig:SAAsim} represent the mean behavior averaged over 30 independent simulations. In the majority of instances, the agent’s trajectory overlaps significantly with the theoretical path, converging toward an identical terminal price. The small difference between two trajectories in Figure ~\ref{fig:SAAsim} comes from the stochastic variance in some LLM's outputs. This strong alignment is expected, as the \emph{Myopic-Profit Agent}  and the \emph{Straightforward Bidding Strategy} share the same objective function:  profit maximization in the immediate subsequent round.  

Second, the \emph{Rule-Centric Agent} demonstrates significantly more aggressive bidding behavior compared to both the tailored agents and the benchmark. In the absence of strategic Chain-of-Thought (CoT) guidance, the \emph{Rule-Centric Agent} consistently places bids exceeding the average of the other agents and the \emph{Straightforward Bidding Strategy}. Notably, Fig.~\ref{fig:SAAsim}(c) captures specific instances of irrational escalation, where the agent continues to bid even after competitors have withdrawn. Analysis of the agent's journals reveal that this behavior stems from a misinterpreted drive to "secure" the product, prioritizing winning over cost-efficiency.

Third, the \emph{Strategic-Outcome Agent} exhibits a more aggressive bidding profile relative to both the \emph{Myopic-Profit Agent} and the \emph{Straightforward Bidding Strategy}. This increased intensity is strategically motivated: the agent utilizes higher bids in earlier rounds to secure a dominant position for its desired products in the final allocation phase. However, the \emph{Strategic-Outcome Agent} remains more conservative than the \emph{Rule-Centric Agent}. This moderation is attributed to the instruction-based ICL modules, which compel the agent to maximize the terminal utility rather than winning at any cost. By framing the objective around final profitability, the ICL module successfully mitigates the irrational bidding escalations observed in the \emph{Rule-Centric Agent}. These findings confirm the efficacy of the ICL framework in maintaining strategic consistency; the agents adhered to their prescribed behavioral personas throughout the entire auction process without lapsing into erratic or stochastic bidding patterns.

% Second, without any hints, the \emph{Rule-Centric Agent} is bidding more aggressively compared to other agents and the math model benchmark (Straightforward Bidding Strategy). As shown in Fig.~\ref{fig:SAAsim}, all of the average bids from the \emph{Rule-Centric Agent} are consistently higher than the average bids from other agents and model. Furthermore, as shown in Fig.~\ref{fig:SAAsim}(c), the \emph{Rule-Centric Agent} shows cases where the agent bids after the other bidder drops out of the auction. This decision was made to maintain the lead from the other bidders and secure the remaining product with the best profit available to oneself, according to the agent's history and journal entries. Such aggressive bidding behavior from the agent implies that LLM agents are capable of simulating not only the optimal mathematical strategies but also the suboptimal, decision-driven strategies, which is closer to the cases in the real world.  

% Third, the \emph{Myopic-Profit Agent} consistently bids closer to the math model, whereas the \emph{Strategic-Outcome Agent} bids closer to the \emph{Rule-Centric Agent}'s bids. Such results show that myopic short-term bidding behaviors face more constraints in their decision process than strategic long-term bidding behaviors do. The results also confirm the power of the ICL modules used in the agents, showing consistency in their bidding strategies throughout the auction without any signs of erratic bidding. 

\begin{figure}
    \centering
    \includegraphics[width=1\linewidth]{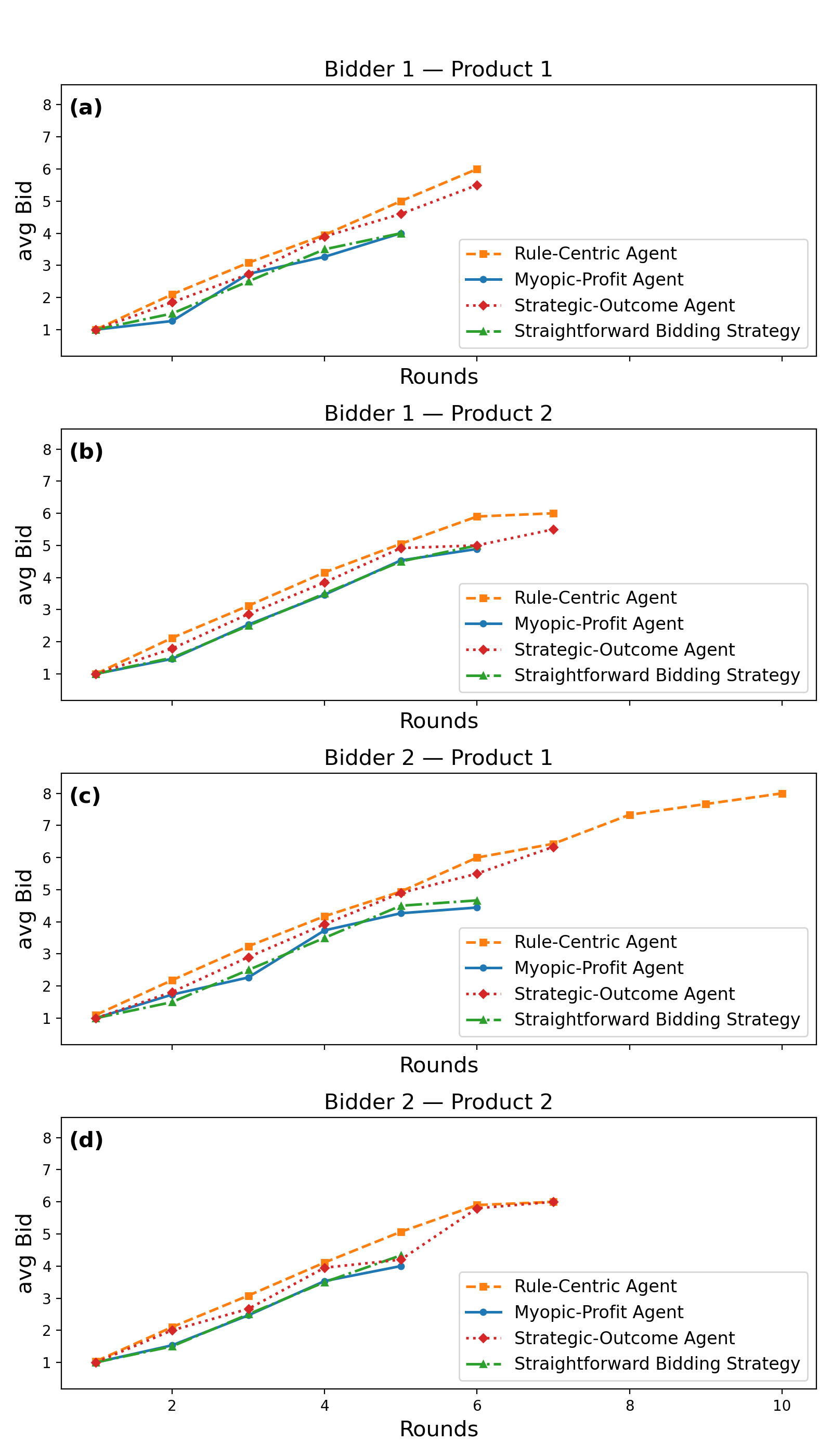}
    \caption{Comparing bidding process/behavior of generative agents and \emph{Straightforward Bidding Strategy}}
    \label{fig:SAAsim}
\end{figure}

\section{Conclusion}

This paper presents two proof-of-concept experiments to explore (1) the effectiveness of in-context learning (ICL) for shaping LLM-based behavioral generative agent decisions, and (2) the ability of LLMs to represent diverse decision-making personas and processes.    

Our results show that LLM-based agents can provide useful insights at both macro and micro levels. At the macro level, simulating multiple LLM agents helps identify different bidding strategies in power auctions that can inform electricity pricing and energy policy. We find the bidding trajectory of \emph{Myopic-Profit Agent} closely replicates that of the \emph{Straightforward Bidding Strategy}, with its average bids converging to nearly the same terminal values. The \emph{Strategic-Outcome Agent} exhibits a more aggressive bidding profile compared to both the \emph{Myopic-Profit Agent} and the \emph{Straightforward Bidding Strategy}. This behavior is driven by a long-term strategic orientation: the agent proactively utilizes higher bids in earlier rounds to establish a dominant market position. Despite this assertiveness, the agent remains more conservative than the \emph{Rule-Centric Agent} because while the \emph{Rule-Centric Agent} may pursue winning at any cost, the \emph{Strategic-Outcome Agent}’s behavior is moderated by a fundamental requirement for profitability. Together, these results highlight both the diversity of LLM agents and the effectiveness of the ICL modules in maintaining coherent, non-erratic strategies.

At the micro level, LLM agents can support personalized decisions in smart home energy systems—for example, guiding battery charging and discharging to balance energy reserves and economic benefits under net energy metering. Furthermore, we find that ICL is effective at transferring behavioral patterns from a stronger model to a smaller LLM, shaping decisions in rare-event scenarios such as blackouts.

% The results of the SAA simulation with different LLM agents further show how capable the agents are in simulating different behaviors with different contexts and objectives. Figure ~\ref{fig:SAAsim} shows that the \emph{Myopic-Profit Agent} closely replicates the \emph{Straightforward Bidding Strategy}, with its average bids converging to nearly the same terminal values and, in some cases, following the mathematical bidding trajectory almost exactly, demonstrating that LLM agents can reproduce optimal analytical strategies. In contrast, the \emph{Rule-Centric Agent} exhibits consistently more aggressive bidding than both the benchmark and other agents, continuing to bid after competitors exit in order to secure the item in some cases, which indicates that LLMs can also generate suboptimal, decision-driven behaviors that resemble real-world strategic deviations. Finally, while the \emph{Myopic-Profit Agent} remains conservative like \emph{Straightforward Bidding Strategy}, the \emph{Strategic-Outcome Agent} aligns more closely with the \emph{Rule-Centric Agent}, suggesting that long-term strategic reasoning allows greater flexibility than short-term profit maximization. 

While this study provides successful proof-of-concepts for the deployment of LLM agents, several avenues for further refinement remain. For instance, the current analysis is grounded in a stylized, simplified example designed to evaluate core agent logic. While these examples offer significant insight, the findings have not yet been validated against real-world market datasets or large-scale network configurations. Furthermore, ICL modules produce only modest differentiation between personas in embedding analyses at the mirco level. To resolve these limitations, we can further integrate data center withdrawal capacity constraints and local grid reliability requirements into the agents' decision-making frameworks to simulate complex real-world systems in the future research. Moreover, we can also explore the eligibility rules in SAA for the agents' prompts to allow for a more robust evaluation of the agents' ability. 

Overall, our study shows that LLM-based behavioral agents can bridge theoretical models and practical energy decision-making, while highlighting areas for future refinement.

% \textcolor{red}{[Cong: explain also the limitations of these two proof of concept experiments in the conclusion. If fig.3 takes many spaces, make the front size larger and shorten the y-axis. Why bidder 2 bid Product 1 for more rounds than bidder 1 in the \emph{Rule-Centric Agent} (orange) setting?]}

\section{Acknowledgment}
Generative AI tools were used for language editing, grammar correction, and clarity improvement. The AI tools did not contribute to the technical content, analysis, or scientific conclusions of this paper.

\section{Appendix} 
% \subsection{Eligibility and Activity Rules in SAA}
% To prevent "sniping" (waiting until the end of the auction to bid), we enforce an activity rule based on a quantity index. Each item $k$ is assigned a weight $w_k$ (in points) representing its value or capacity size. Each bidder $n$ begins with an initial eligibility $E_{n,0}$, determined by their initial deposit.

% A bidder $n$ is considered \textit{active} on item $k$ at round $t$ if they either hold the standing high bid from round $t-1$ or submit a valid new bid in the current round. Let $S_{n,t} \subseteq \mathcal{K}$ be the set of items on which bidder $n$ is active. The bidder's activity level $A_{n,t}$ is the sum of weights of these items:
% \begin{equation}
%     A_{n,t} = \sum_{k \in S_{n,t}} w_k.
% \end{equation}
% The auction enforces two constraints. First, a bidder cannot submit bids that cause their activity to exceed their current eligibility: $A_{n,t} \le E_{n,t}$. Second, eligibility is updated strictly based on activity. During auction stage $j$, a bidder must maintain activity on a fraction $f_j \in (0, 1]$ of their eligibility. If the activity $A_{n,t}$ falls below the required threshold $f_j E_{n,t}$, the eligibility for the subsequent round is reduced:
% \begin{equation}
%     E_{n, t+1} = \begin{cases} 
%     E_{n,t} & \text{if } A_{n,t} \ge f_j E_{n,t} \\
%     \frac{A_{n,t}}{f_j} & \text{if } A_{n,t} < f_j E_{n,t}
%     \end{cases}.
% \end{equation}
% This rule ensures that eligibility is monotonic non-increasing, forcing bidders to reveal their demand early.

\subsection{Prompt Designs and Additional Simulation Results}\label{sec:batteryexperiment}

We extend the prompt design in \cite{Chen25behavioral} by incorporating an explicit in-context learning (ICL) module. Specifically, we augment the original persona-based and experiment-specific prompt by inserting example decision trajectories, as illustrated below:

 \begin{tcolorbox}[
  colback=white,           % background color
  colframe=black,          % frame color
  title=In-context learning examples,
  fonttitle=\bfseries,
  fontupper=\small,        % Set the text inside the box to a smaller size
  breakable,
  enhanced,
  sharp corners
]
...
        \textbf{Here are some example responses to guide you:}
        
        \{example\_responses\}
\end{tcolorbox}

The ellipses (“…”) indicate that the remainder of the prompt remains unchanged from \cite{Chen25behavioral}, including the persona definitions, experiment setup, and blackout intervention instructions. The ICL examples are injected as an additional prompt component to guide the LLM’s dispatch decisions without altering the underlying decision framework.

\begin{figure}
    \centering
    \includegraphics[width=0.45\linewidth]{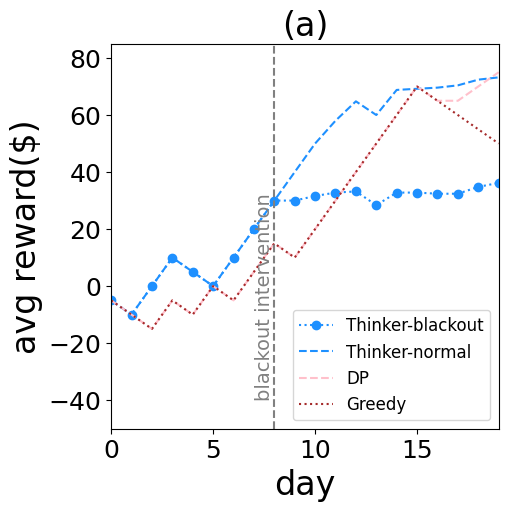}\includegraphics[width=0.45\linewidth]{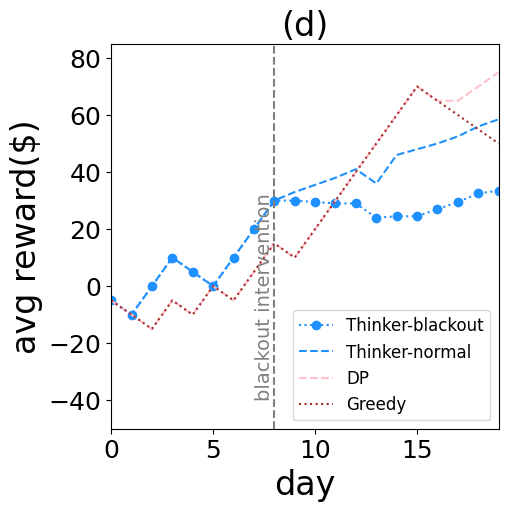}
    \includegraphics[width=0.45\linewidth]{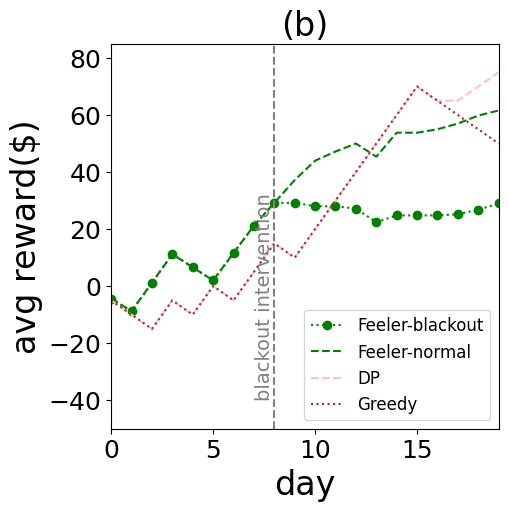}\includegraphics[width=0.45\linewidth]{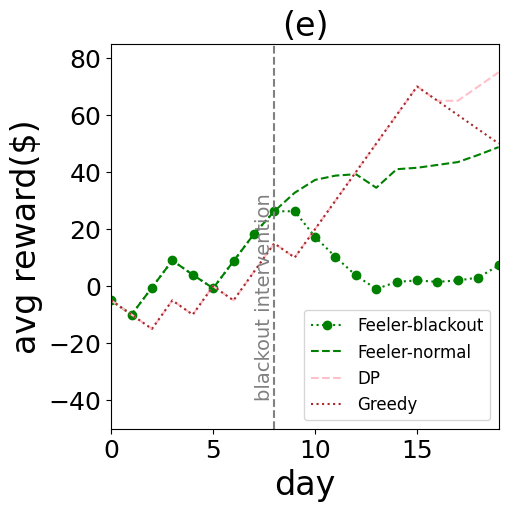}
    \includegraphics[width=0.45\linewidth]{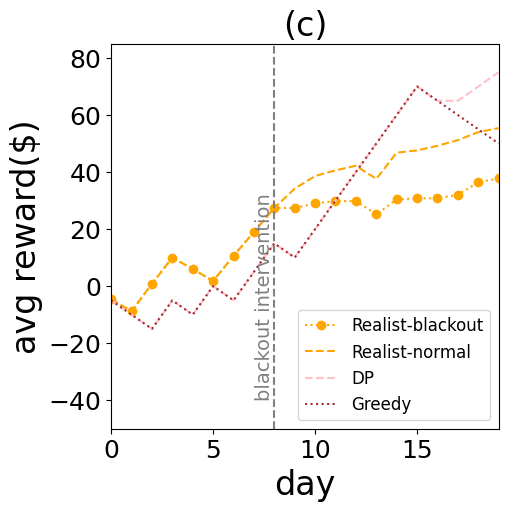}\includegraphics[width=0.45\linewidth]{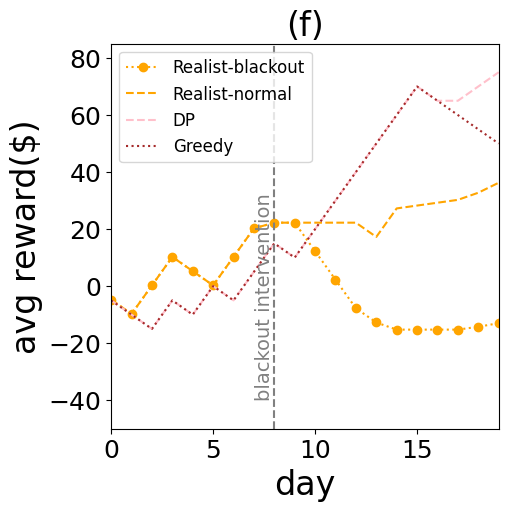}
    \caption{Average accumulated reward with blackout intervention. Left: without IC examples. Right: with ICL-blackout examples.}
    \label{fig:reward}
      \vspace{-2em}
\end{figure}

We provide additional simulation results on terminal rewards in Fig.~\ref{fig:reward}. During blackout interventions, generative agents tend to prioritize energy security and reserve availability over profit-oriented buy-low–sell-high strategies, resulting in lower accumulated rewards compared to scenarios without blackouts. These reward outcomes further reflect the limitation discussed in Sec.~\ref{sec:simbattery}. %In particular, o1-preview typically yields lower terminal rewards and higher terminal SoC for the Feeler persona, whereas gpt-5-mini exhibits the opposite pattern, although the differences are small. This discrepancy highlights the bounds of behavior transfer via in-context learning (ICL), where qualitative decision patterns are preserved but quantitative outcomes may differ.

To facilitate comparison between o1-preview and gpt-5-mini, we reproduce the o1-preview results from \cite{Chen25behavioral} in Fig.~\ref{fig:SoCo1}. As shown in Fig.~\ref{fig:SoCo1}(b–c), the Feeler and Realist personas exhibit higher terminal SoC on Day 20 under blackout interventions, indicated by the red dashed curves. The corresponding textual outputs show that, after experiencing a blackout, these personas deliberately retain more energy in the battery as backup reserves. % This behavioral pattern reappears in the right column of Fig.~\ref{fig:SoCall} when gpt-5-mini is guided by in-context examples generated by o1-preview, demonstrating effective qualitative behavior transfer via in-context learning.
\begin{figure}
    \centering
    \includegraphics[width=0.45\linewidth]{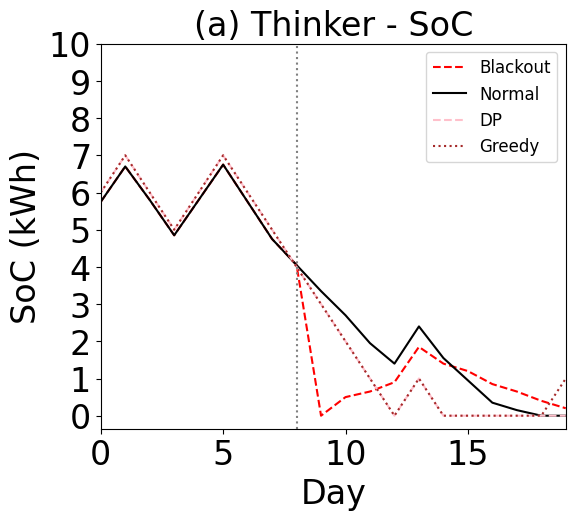}\includegraphics[width=0.45\linewidth]{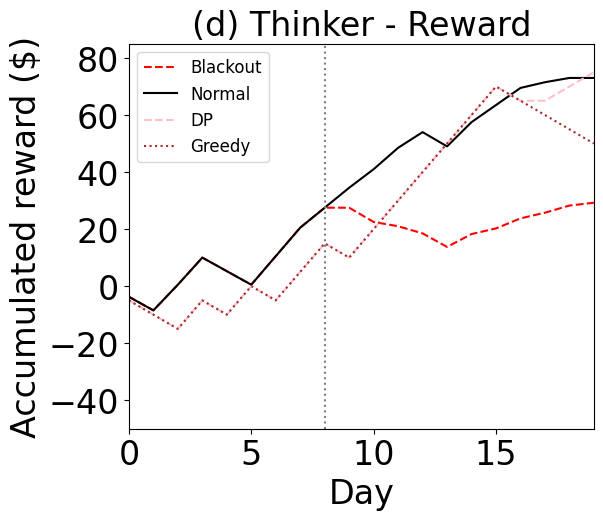}
    \includegraphics[width=0.45\linewidth]{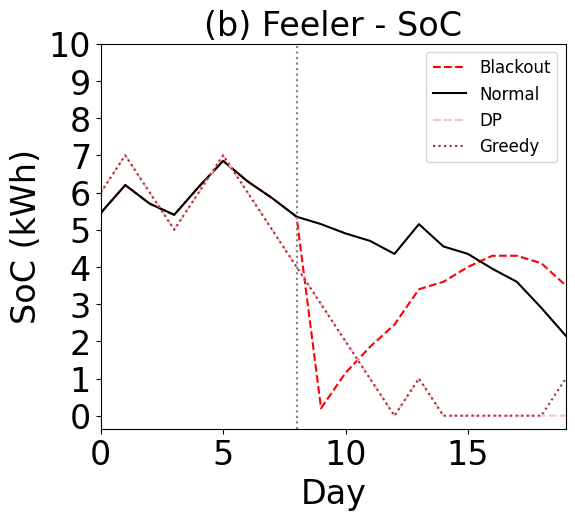}\includegraphics[width=0.45\linewidth]{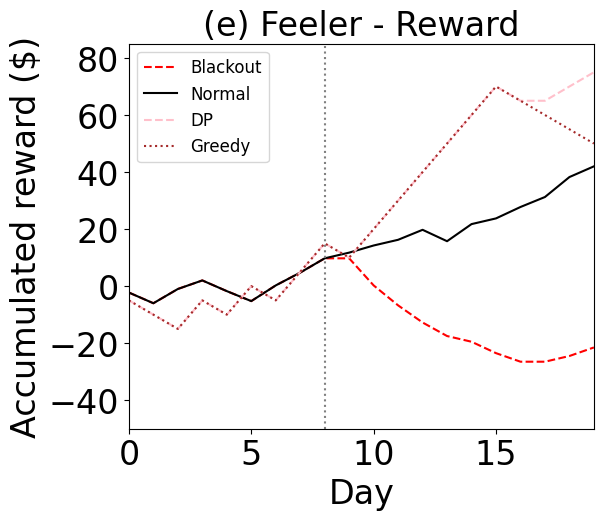}
    \includegraphics[width=0.45\linewidth]{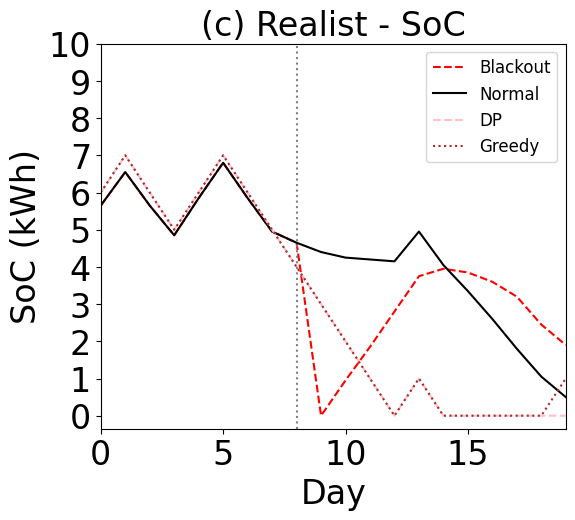}\includegraphics[width=0.45\linewidth]{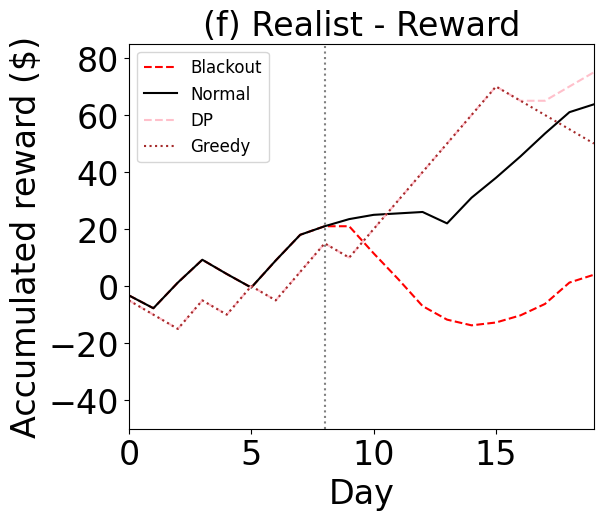}
    \caption{o1-preview results for the average SoC and accumulated reward over the 20-day experiment. \cite{Chen25behavioral}}
    \label{fig:SoCo1}
      \vspace{-2em}
\end{figure} 

\subsection{Prompt Design for SAA}\label{sec:promptdesign}
We design three distinct prompting frameworks to instantiate the LLM agents as bidders in the SAA environment: \emph{Rule-Centric, Myopic-Profit, and Strategic-Outcome}. Each agent is initialized with unique private valuations for individual products and bundles. For \emph{Rule-Centric Agent}, the LLM is tasked with decision-making for each auction round without additional strategic guidance. The prompt provides only the fundamental SAA protocols, requiring the agent to navigate the auction based on its internal reasoning.  For \emph{Myopic-Profit Agent}, in addition to the rules of SAA, we incorporate an instruction-based in-Context Learning (ICL) module designed to simulate mathematically optimized, short-term behavior. This agent is explicitly tasked with maximizing utility for the immediate round. The ICL module conditions the agent to: perform cost-benefit analyses by comparing valuation against cost and choose product combinations for immediate payoff. Moreover, we give detailed instructions for certain cases, such as bidding the current price when the agent owns the product from the previous round. Through these instructions, the LLM treats the prompt as a dynamic operational context, adapting its bidding logic to mimic a myopic optimizer. The \emph{Strategic-Outcome Agent} is designed for long-term profit maximization, prioritizing terminal utility. The instruction-based ICL module includes consideration of current prices and increments and comparison between values and cost, which is the same as the \emph{Myopic-Profit Agent}. Here, however, the module direct the agent to evaluate how current bids influence future market states, price trajectories, and product availability. Notably, this module is given more strategic latitude; unlike the Myopic agent, it is not restricted by specific bidding instructions for bidding behavior.

 \begin{tcolorbox}[
  colback=white,           % background color
  colframe=black,          % frame color
  title=Prompt for \emph{Rule-Centric Agent},
  fonttitle=\bfseries,
  fontupper=\small,        % Set the text inside the box to a smaller size
  breakable,
  enhanced,
  sharp corners
]

You are bidder {bidder} in a SAA.

======================  \\
\textbf{Current Auction State}   \\
======================  \\
Products available: \{products\_str\}   \\
Current prices: \{prices\}              \\
Current high bidders: \{high\_bidders\} \\
Minimum allowed price increase for challenges: \{min\_inc\}
======================  \\
\textbf{YOUR PAST EXPERIENCE}    \\
======================  \\
Here is a short summary of what you have done so far in previous rounds: \{history\_text\}

Here is your private journal from earlier rounds: \{journal\_text\}

Use this information to stay consistent with your own past strategy and learning.

======================  \\
\textbf{YOUR VALUATIONS}         \\
======================  \\
Below is a table showing how much you value each individual product and, when available, some bundles of products together.

Keys that look like ["Product 1", "Product 2"] are bundle values for owning that combination of items.

\{val\_json\_str\}

======================  \\
\textbf{Auction Rule}            \\
======================  \\
The auction proceeds in rounds. In each round, every bidder simultaneously submits bids on any subset of products. 

For each product, there is a current price \{prices\} for products and a current provisional high bidder \{high\_bidders\}.
    
In each round, each bidder submits a bid for a product.

Whoever bids the highest bid becomes the provisional high bidder and claims the product, and the highest bid becomes the new price of the product for the next round. 

If two bidders submit the same highest valid bid, the winner is chosen randomly (fair tie-breaking). 

Products with no valid bids retain their current price and high bidder.
    
To enter the auction in round 1, you MUST pay at least the minimum increment to the initial price given. 
The auction continues until a round occurs where:
\begin{itemize}
    \item No prices change, and
    \item No high bidders change across all products.
\end{itemize}
When this happens, the auction terminates.

======================  \\
\textbf{Your Role}               \\
======================  \\
For each round, you have to decide whether to bid or not for a given product. 
If you decide to bid, then you have to decide how much to bid for each product based on the valuation table \{val\_json\_str\} and history \{history\_text\}. 

======================      \\
TARJ OUTPUT FORMAT (STRICT) \\
======================      \\
Report your internal reasoning and your decision using the following fields, in order:

\textbf{Thoughts:} (In at most 50 words, describe your internal reasoning, how you compared different combinations, and why you chose the final combination.)

\textbf{Action:} (In at most 50 words, describe in natural language what you are doing this round: which combination you are effectively demanding and how you are bidding.)

ChosenSubset: \{\{comma-separated product names in the combination you chose;
leave empty if you decided to demand nothing\}\}
\{bid\_lines\}

\textbf{Reflection:} (In at most 50 words, reflect on whether this decision fits your overall goals and how the current prices and standings affect your strategy.)

\textbf{Journal:} (In at most 50 words, write a short journal entry to your future self about what you did this round and what you might watch for in future rounds.)
\end{tcolorbox}

% \textcolor{red}{[Cong: if some prompt parts are the same, you can just use the ellipses (“…”) and explain that the omitted part indicates prompt unchanged compared with the rule-centric agent.]}

\begin{tcolorbox}[
  colback=white,
  colframe=black,
  title=Prompt for \emph{Myopic-Profit Agent},
  fonttitle=\bfseries,
  fontupper=\small,
  breakable,
  enhanced,
  sharp corners
]
\textit{(Note: The omitted parts [“...”] indicate prompt content that remains unchanged compared to the \emph{Rule-Centric Agent} prompt. We omit them due to page limitations.)}

You are bidder \{bidder\} in a SAA.

\textbf{... }

{Current Auction State, Past Experience, and Valuations sections are identical to the \emph{Rule-Centric Agent}} 

\textbf{... }

======================  \\
HOW YOU SHOULD DECIDE   \\
======================  \\
In this round, you must decide which combination of products you want to try to end up with and how you want to bid on each product.

Think carefully about:
\begin{itemize}
    \item All the reasonable combinations of products you could aim for (including possibly choosing no products at all).
    \item For each combination, how much total value it has for you, using the table \{val\_json\_str\} 
    (use the bundle value if it is explicitly given; otherwise add up the values of the individual products in that combination).
    \item What you would likely have to pay for that combination, considering:
    \begin{itemize}
        \item the current price of each product,
        \item whether you are already the provisional winner on that product,
        \item whether someone else is currently winning and you would need to raise the price by at least the minimum increment to challenge.
    \end{itemize}
    \item The resulting utility for each combination (value minus the cost you would reasonably expect to pay).
\end{itemize}

Choose the combination that gives you the highest utility and is reasonable given your overall strategy from previous rounds. 

If every combination seems worse than doing nothing, you may choose to demand no products this round.

If several combinations have the same utility, it is acceptable to prefer a combination with more products. 

If this is the first round, you MUST pay at least the minimum increment to the initial price given to start and stay in the auction. 

Once you decide which combination you are aiming for in this round:
\begin{itemize}
    \item For each product that is NOT part of your chosen combination, you should not bid on it (output "none" for that product).
    \item For each product that IS part of your chosen combination:
        \begin{itemize}
            \item If no one is currently the high bidder, it is usually enough to bid at the current price.
            \item If you are already the high bidder, you can usually keep it by bidding at the current price again.
            \item If another bidder is currently the high bidder, you must be willing to offer at least the current price plus the minimum increment to challenge them.
        \end{itemize}
\end{itemize}

\textbf{... [TARJ Output Format is identical to the \emph{Rule-Centric Agent}] ...}

\end{tcolorbox}

% ==========================================
% BOX 3: STRATEGIC-OUTCOME AGENT (ABBREVIATED)
% ==========================================
\begin{tcolorbox}[
  colback=white,
  colframe=black,
  title=Prompt for \emph{Strategic-Outcome Agent},
  fonttitle=\bfseries,
  fontupper=\small,
  breakable,
  enhanced,
  sharp corners
]
\textit{(Note: The omitted parts [“...”] indicate prompt content that remains unchanged compared to the \emph{Rule-Centric Agent} prompt. We omit them due to page limitation.)}

You are bidder \{bidder\} in a SAA.

\textbf{... }

{Current Auction State, Past Experience, and Valuations sections are identical to the \emph{Rule-Centric Agent}} 

\textbf{... }

======================  \\
HOW YOU SHOULD DECIDE   \\
======================  \\
In this round, you must decide which combination of products you want to try to end up with and how you want to bid on each product.

Think carefully about:
\begin{itemize}
    \item All the reasonable combinations of products you could aim for (including possibly choosing no products at all).
    \item For each combination, how much total value it has for you, using the table \{val\_json\_str\} 
    (use the bundle value if it is explicitly given; otherwise add up the values of the individual products in that combination).
    \item What you would likely have to pay for that combination, considering:
    \begin{itemize}
        \item the current price of each product,
        \item whether you are already the provisional winner on that product,
        \item whether someone else is currently winning and you would need to raise the price by at least the minimum increment to challenge.
    \end{itemize}
    \item The resulting utility for each combination (value minus the cost you would reasonably expect to pay).
    \item How the current combination will affect the future combinations, prices, and utility in the upcoming rounds.
\end{itemize}

Choose the combination that gives you the highest utility and is reasonable given your overall strategy from previous rounds. 

If every combination seems worse than doing nothing, you may choose to demand no products this round.

If several combinations have the same utility, it is acceptable to prefer a combination with more products. 

If this is the first round, you MUST pay at least the minimum increment to the initial price given to start and stay in the auction. 

Once you decide which combination you are aiming for in this round:
\begin{itemize}
    \item For each product that is NOT part of your chosen combination, you should not bid on it (output "none" for that product).
    \item For each product that IS part of your chosen combination:
        \begin{itemize}
            \item If another bidder is currently the high bidder, you must be willing to offer at least the current price plus the minimum increment to challenge them.
        \end{itemize}
\end{itemize}

\textbf{... [TARJ Output Format is identical to the \emph{Rule-Centric Agent}] ...}

\end{tcolorbox}

\bibliographystyle{ieeetr}
\bibliography{references}

\begin{thebibliography}{10}

\bibitem{park25GenerativeAgents}
J.~S. Park, {\em Generative Agent Simulations of Human Behavior}.
\newblock PhD thesis, Stanford University, 2025.
\newblock ProQuest Dissertations \& Theses, No. 32316468.

\bibitem{horton2023large}
J.~J. Horton, ``Large language models as simulated economic agents: What can we learn from homo silicus?,'' tech. rep., National Bureau of Economic Research, 2023.

\bibitem{aher2023using}
G.~V. Aher, R.~I. Arriaga, and A.~T. Kalai, ``Using large language models to simulate multiple humans and replicate human subject studies,'' in {\em International conference on machine learning}, pp.~337--371, PMLR, 2023.

\bibitem{ArgyleEtAl2023OutOfOneMany}
L.~P. Argyle, E.~C. Busby, N.~Fulda, J.~Gubler, C.~Rytting, and D.~Wingate, ``Out of one, many: Using language models to simulate human samples,'' {\em Political Analysis}, vol.~31, no.~3, pp.~337--351, 2023.

\bibitem{RossKimLo2024LLMEconomicus}
J.~Ross, Y.~Kim, and A.~W. Lo, ``{LLM} economicus? mapping the behavioral biases of {LLM}s via utility theory,'' {\em arXiv preprint}, 2024.

\bibitem{Leng2024MentalAccounting}
Y.~Leng, ``Can {LLM}s mimic human-like mental accounting and behavioral biases?,'' {\em SSRN Working Paper}, 2024.
\newblock SSRN 4705130.

\bibitem{akata2025playing}
E.~Akata, L.~Schulz, J.~Coda-Forno, S.~J. Oh, M.~Bethge, and E.~Schulz, ``Playing repeated games with large language models,'' {\em Nature Human Behaviour}, pp.~1--11, 2025.

\bibitem{fontana2025nicer}
N.~Fontana, F.~Pierri, and L.~M. Aiello, ``Nicer than humans: How do large language models behave in the prisoner's dilemma?,'' in {\em Proceedings of the International AAAI Conference on Web and Social Media}, vol.~19, pp.~522--535, 2025.

\bibitem{duan2024gtbench}
J.~Duan, R.~Zhang, J.~Diffenderfer, B.~Kailkhura, L.~Sun, E.~Stengel-Eskin, M.~Bansal, T.~Chen, and K.~Xu, ``Gtbench: Uncovering the strategic reasoning capabilities of llms via game-theoretic evaluations,'' {\em Advances in Neural Information Processing Systems}, vol.~37, pp.~28219--28253, 2024.

\bibitem{GoodyearGuoJohari2025StateRep}
L.~Goodyear, R.~Guo, and R.~Johari, ``The effect of state representation on {LLM} agent behavior in dynamic routing games,'' {\em arXiv preprint}, 2025.

\bibitem{brown2020language}
T.~Brown, B.~Mann, N.~Ryder, M.~Subbiah, J.~D. Kaplan, P.~Dhariwal, A.~Neelakantan, P.~Shyam, G.~Sastry, A.~Askell, {\em et~al.}, ``Language models are few-shot learners,'' {\em Advances in neural information processing systems}, vol.~33, pp.~1877--1901, 2020.

\bibitem{wei2022chain}
J.~Wei, X.~Wang, D.~Schuurmans, M.~Bosma, F.~Xia, E.~Chi, Q.~V. Le, D.~Zhou, {\em et~al.}, ``Chain-of-thought prompting elicits reasoning in large language models,'' {\em Advances in neural information processing systems}, vol.~35, pp.~24824--24837, 2022.

\bibitem{AgrawalYucel2022DRPrograms}
V.~V. Agrawal and {\c{S}}.~Y{\"u}cel, ``Design of electricity demand-response programs,'' {\em Management Science}, vol.~68, no.~10, pp.~7441--7456, 2022.

\bibitem{GaoAlshehriBirge2024DERMarketPower}
Z.~Gao, K.~Alshehri, and J.~R. Birge, ``Aggregating distributed energy resources: Efficiency and market power,'' {\em Manufacturing \& Service Operations Management}, vol.~26, no.~3, pp.~834--852, 2024.

\bibitem{Chen25behavioral}
C.~Chen, O.~Karaduman, and X.~Kuang, ``Behavioral generative agents for energy operations,'' {\em arXiv preprint arXiv:2506.12664}, 2025.

\bibitem{FilippasHortonManning2024HomoSilicusEC}
A.~Filippas, J.~J. Horton, and B.~S. Manning, ``Large language models as simulated economic agents: What can we learn from homo silicus?,'' in {\em Proceedings of the 25th ACM Conference on Economics and Computation (EC)}, pp.~614--615, 2024.

\bibitem{ShanahanMcDonellReynolds2023RolePlay}
M.~Shanahan, K.~McDonell, and L.~Reynolds, ``Role play with large language models,'' {\em Nature}, vol.~623, no.~7987, pp.~493--498, 2023.

\bibitem{bertsekas2012dynamic}
D.~Bertsekas, {\em Dynamic programming and optimal control: Volume I}, vol.~4.
\newblock Athena scientific, 2012.

\bibitem{JiangPowell15Storage}
D.~R. Jiang and W.~B. Powell, ``Optimal hour-ahead bidding in the real-time electricity market with battery storage using approximate dynamic programming,'' {\em INFORMS Journal on Computing}, vol.~27, no.~3, pp.~525--543, 2015.

\end{thebibliography}
\balance

\endgroup
\end{document}